\documentclass[12pt]{article}

\usepackage{cite}
\usepackage{booktabs}
\usepackage{hyperref}
\usepackage{graphicx}
\usepackage{subfig}
\usepackage{paralist}

\usepackage{geometry}
\usepackage{amsthm}
\usepackage{amsmath}
\usepackage{amssymb}

\allowdisplaybreaks[1]

\newtheorem{theorem}{Theorem}
\newtheorem{lemma}[theorem]{Lemma}

\makeatletter
\def\@endtheorem{\endtrivlist}
\makeatother

\newcommand{\BP}{\mathit{BP}}
\newcommand{\NULL}{\mathrm{NULL}}

\newcommand{\Labelb}{\mathrm{label}}
\newcommand{\depthb}{\mathrm{depth}}
\newcommand{\parentb}{\mathrm{parent}}
\newcommand{\levelancestorb}{\mathrm{level\_ancestor}}
\newcommand{\Degb}{\mathrm{deg}}
\newcommand{\childrankb}{\mathrm{child\_rank}}
\newcommand{\childselectb}{\mathrm{child\_select}}
\newcommand{\numdescendantsb}{\mathrm{num\_descendants}}
\newcommand{\preorderrankb}{\mathrm{preorder\_rank}}
\newcommand{\preorderselectb}{\mathrm{preorder\_select}}
\newcommand{\postorderrankb}{\mathrm{postorder\_rank}}
\newcommand{\postorderselectb}{\mathrm{postorder\_select}}

\newcommand{\lcab}{\mathrm{lca}}
\newcommand{\bprankb}{\mathrm{bp\_rank}}
\newcommand{\bpselectb}{\mathrm{bp\_select}}
\newcommand{\bpopenb}{\mathrm{bp\_open}}
\newcommand{\bpcloseb}{\mathrm{bp\_close}}
\newcommand{\bpnodeb}{\mathrm{bp\_node}}

\newcommand{\Label}[1]{\Labelb(#1)}
\newcommand{\depth}[2]{\depthb_{#1}(#2)}
\newcommand{\parent}[2]{\parentb_{#1}(#2)}
\newcommand{\levelancestor}[3]{\levelancestorb_{#1}(#2,#3)}
\newcommand{\Deg}[2]{\Degb_{#1}(#2)}
\newcommand{\childrank}[2]{\childrankb_{#1}(#2)}
\newcommand{\childselect}[3]{\childselectb_{#1}(#2,#3)}
\newcommand{\numdescendants}[2]{\numdescendantsb_{#1}(#2)}
\newcommand{\preorderrank}[2]{\preorderrankb_{#1}(#2)}
\newcommand{\preorderselect}[2]{\preorderselectb_{#1}(#2)}
\newcommand{\postorderrank}[2]{\postorderrankb_{#1}(#2)}
\newcommand{\postorderselect}[2]{\postorderselectb_{#1}(#2)}

\newcommand{\parentx}[1]{\parentb(#1)}

\newcommand{\preorderrankx}[1]{\preorderrankb(#1)}
\newcommand{\preorderselectx}[1]{\preorderselectb(#1)}

\newcommand{\lcax}[2]{\lcab(#1,#2)}

\newcommand{\D}[2]{\mathcal{D}_{#1,#2}} 
\newcommand{\T}[2]{#1_{#2}} 
\newcommand{\V}[1]{V(#1)} 
\newcommand{\C}[1]{C(#1)} 

\newcommand{\idxzero}{0}
\newcommand{\idxtotal}{1}
\newcommand{\idxle}{2}
\newcommand{\idxgt}{3}
\newcommand{\idxeq}{4}
\newcommand{\idxchildren}{5}
\newcommand{\idxleftmostchild}{6}
\newcommand{\idxspecial}{7}

\newcommand{\Sumb}{\mathrm{sum}}
\newcommand{\fwdsearchb}{\mathrm{fwd\_search}}
\newcommand{\bwdsearchb}{\mathrm{bwd\_search}}
\newcommand{\RMQib}{\mathrm{RMQi}}

\newcommand{\bp}[1]{\mathrm{bp}(#1)}
\newcommand{\bprank}[2]{\bprankb_{#1}(#2)}
\newcommand{\bpselect}[2]{\bpselectb_{#1}(#2)}
\newcommand{\bpopen}[1]{\bpopenb(#1)}
\newcommand{\bpclose}[1]{\bpcloseb(#1)}
\newcommand{\bpnode}[1]{\bpnodeb(#1)}

\newcommand{\Sum}[4]{\Sumb(#1,#2,#3,#4)}
\newcommand{\fwdsearch}[4]{\fwdsearchb(#1,#2,#3,#4)}
\newcommand{\bwdsearch}[4]{\bwdsearchb(#1,#2,#3,#4)}
\newcommand{\RMQi}[4]{\RMQib(#1,#2,#3,#4)}

\newcommand{\rank}[3]{\mathrm{rank}_{#1}(#2,#3)}
\newcommand{\select}[3]{\mathrm{select}_{#1}(#2,#3)}
\newcommand{\access}[2]{\mathrm{access}(#1,#2)}
\newcommand{\tselect}{t_{\mathrm{select}}}

\newcommand{\map}[1]{\mathrm{map}(#1)}
\newcommand{\range}[2]{\{#1,\ldots,#2\}}

\begin{document}

\title{Succinct representation of labeled trees}
\author{Dekel Tsur%
\thanks{Department of Computer Science, Ben-Gurion University of the Negev.
Email: \texttt{dekelts@cs.bgu.ac.il}}}
\date{}
\maketitle

\begin{abstract}
We give a representation for labeled ordered trees
that supports labeled queries such as finding the $i$-th ancestor of a node
with a given label.
Our representation is succinct, namely the redundancy is
small-o of the optimal space for storing the tree.
This improves the representation of He et al.~\cite{HeMZ12} which is
succinct unless the entropy of the labels is small.
\end{abstract}

\section{Introduction}

A problem which was extensively studied in recent years is representing an ordered rooted
tree using space close to the information-theoretic lower bound
while supporting numerous queries on the tree, e.g.~\cite{Jacobson89,MunroR01,ChiangLL05,MunroRRR12,BenoitDMRRR05,GearyRR06,FarzanM08,NavarroS14}.
Geary et al.~\cite{GearyRR06} studied an extension of this problem, in which
the nodes of the tree are labeled with characters from alphabet $\range{1}{\sigma}$.
The tree queries now receive a character $\alpha$ as an additional argument, and the
goal of a query is to locate a certain node whose label is $\alpha$ or to count
the nodes satisfying some property and whose labels are $\alpha$.
The set of queries considered by Geary et al.\ is given in
Table~\ref{tab:labeled-queries}.
Geary et al.\ gave a representation that uses
$n\log\sigma+2n+O(n\sigma\log\log\log n/\log \log n)$ bits,
where $n$ is the number of nodes and $\sigma$ is the size of the alphabet,
and answers queries in constant time.
For $\sigma=o(\log\log n/\log\log\log n)$, the space is
$n\log\sigma+2n+o(n)$, namely, the space is $o(n)$ more than the
information-theoretic lower bound.

Barbay et al.~\cite{BarbayGMR07} and Ferragina et al.~\cite{FerraginaLMM09}
gave labeled tree representations that use space close to the lower bound
for large alphabets,
but the set of supported queries is more restricted.
He et al.~\cite{HeMZ12} improved the result of Geary et al.\ by showing
a labeled tree representation based upon a rank-select structure on
the string $P_T$ that contains the labels of the nodes in preorder.
Using the rank-select structure of Belazzougui and Navarro~\cite{BelazzouguiN12}
the following results were obtained:
\begin{inparaenum}[(1)]
\item
For $\sigma=w^{O(1)}$, there is a representation that uses $n H_0(P_T)+O(n)$
bits and answers queries in $O(1)$ time,
where $w$ is the word size and
$H_0(P_T)$ is the zero-order entropy of $P_T$.
\item
For $\sigma\leq n$, there is a representation that uses
$n H_0(P_T)+o(n H_0(P_T))+O(n)$ bits.
$\Labelb$ queries are answered in $O(1)$ time,
and $\preorderrankb$ queries are answered in $\omega(1)$ time
(namely, for every function $f$ satisfying $f(n)=\omega(1)$, the time is $O(f(n))$), or vice versa. Other queries are answered in
$O(\log\frac{\log\sigma}{\log w})$ time.
\end{inparaenum}
The representation of He et al.\ supports all the queries of
Table~\ref{tab:labeled-queries} and additional queries.
Note that the representation is succinct if $H_0(P_T)=\Omega(1)$.

In this paper we give a fully succinct representation of labeled trees.
Our result is given in the following theorem.
\begin{theorem}
Let $T$ be a labeled tree with $n$ nodes and labels from $\range{1}{\sigma}$.
\begin{enumerate}
\item
For $\sigma=w^{O(1)}$, there is a representation of $T$ that uses
$n H_0(P_T)+2n+o(n)$ bits
and answers the queries of Table~\ref{tab:labeled-queries} in $\omega(1)$ time.
\item
For $\sigma\leq n$, there is a representation of $T$ that
uses $n H_0(P_T)+2n+o(n H_0(P_T))+O(n)$ bits.
$\Labelb$ queries are answered in $O(1)$ time,
and $\preorderrankb$ queries are answered in $\omega(1)$ time,
or vice versa.
The rest of the queries of Table~\ref{tab:labeled-queries} are answered in
$O(\log\frac{\log\sigma}{\log w})$ time.
\end{enumerate}
\end{theorem}
Note that our representation supports only the queries considered by
Geary et al.\ and it does not support the additional queries considered by He et al.

\begin{table}
\caption{Supported queries on a labeled tree.
A node with label $\alpha$ is called an $\alpha$-node.
A $\alpha$-child of a node is a child which is an $\alpha$-node.
Other $\alpha$- terms are defined similarly.
\label{tab:labeled-queries}}
\centering
\begin{tabular}{lp{11.4cm}}
\toprule
Query & Description \\
\midrule
$\Label{x}$ & The label of $x$. \\
$\depth{\alpha}{x}$ & The number of $\alpha$-nodes on the path from
						the root to $x$. \\
$\levelancestor{\alpha}{x}{i}$ & The $\alpha$-ancestor $y$ of $x$ for which
						$\depth{\alpha}{y} = \depth{\alpha}{x}-i$. \\
$\parent{\alpha}{x}$ & $\levelancestor{\alpha}{x}{1}$. \\
$\Deg{\alpha}{x}$ & The number of $\alpha$-children of $x$. \\
$\childrank{\alpha}{x}$ & The rank of $x$ among its $\alpha$-siblings. \\
$\childselect{\alpha}{x}{i}$ & The $i$-th $\alpha$-child of $x$. \\
$\numdescendants{\alpha}{x}$ & The number of $\alpha$-descendants of $x$. \\
$\preorderrank{\alpha}{x}$ & The preorder rank of $x$ among
the $\alpha$-nodes. \\
$\preorderselect{\alpha}{i}$ & The $i$-th $\alpha$-node in the preorder. \\
$\postorderrank{\alpha}{x}$ & The postorder rank of $x$ among
the $\alpha$-nodes. \\
$\postorderselect{\alpha}{i}$ & The $i$-th $\alpha$-node in the postorder. \\
\bottomrule
\end{tabular}
\end{table}

\section{Preliminaries}
A \emph{rank-select structure} stores a string $S$ over alphabet
$\range{1}{\sigma}$ and supports the following queries:
\begin{inparaenum}[(1)]
\item
$\rank{\alpha}{S}{i}$ returns the number of occurrences of $\alpha$ in the first
$i$ characters of $S$
\item
$\select{\alpha}{S}{i}$ returns the $i$-th occurrence of $\alpha$ in $S$
\item
$\access{S}{i}$ returns the $i$-th character of $S$.
\end{inparaenum}
The problem of designing a succinct rank-select structure with efficient query
times was studied extensively. For our purpose, we use the following results.
\begin{theorem}[Belazzougui and Navarro~\cite{BelazzouguiN12}]
\label{thm:rank-select}
A rank-select structure can be built on a string $S$ of length $n$
over alphabet $\range{1}{\sigma}$ such that
\begin{inparaenum}[(1)]
\item
If $\sigma=w^{O(1)}$, the space is $n H_0(S)+o(n)$ bits,
and the structure answers rank queries in $O(1)$ time
\item
If $\sigma \leq n$, the space is $n H_0(S)+o(n H_0(S))+o(n)$ bits,
and the structure answers rank queries
in $O(\log \frac{\log\sigma}{\log w})$ time.
The structure answers access queries in $O(1)$ time and select
queries in $\omega(1)$ time, or vice versa.
\end{inparaenum}
\end{theorem}
\begin{theorem}[Raman et al.~\cite{RamanRS07}]
\label{thm:rank-select-binary}
A rank-select structure can be built on a binary string $S$ of length $n$
that contains $k$ ones
such that the space is $O(k \log(n/k))+o(n)$ bits, and
the structure answers queries in $O(1)$ time.
\end{theorem}

We also use the following result on representation of (unlabeled) ordered trees.
We are intrested in a representation that supports
the unlabeled versions of the queries listed in Table~\ref{tab:labeled-queries}
and additional queries such as $\lcab$ queries.
\begin{theorem}[Navarro and Sadakane~\cite{NavarroS14}]
\label{thm:unlabeled-tree}
An ordered tree can be stored using $2n+o(n)$ bits such that
tree queries can be answered in $O(1)$ time.
\end{theorem}

\subsection{Representation of weighted trees}
In this section we consider the problem of representing ordered trees with
weights on the nodes. We will use weighted trees in our representation of
labeled trees.

Let $T$ be a tree with weights $w_1(v),\ldots, w_s(v)$ for each node,
where each weight is from $\range{0}{X-1}$.
For a set of nodes $U$, the \emph{$w_a$-weight} of $U$ is
$\sum_{v\in U} w_a(v)$.
The weighted tree queries we need are described in
Table~\ref{tab:weight-queries}. Throughout this section we assume that the balanced
parenthesis string of a tree is a binary string, where open and close parenthesis are
represented by 1 and 0, respectively.

\begin{table}
\caption{Supported queries on a weighted tree.
$\BP$ denotes the balanced parenthesis representation of the tree.
\label{tab:weight-queries}}
\centering
\begin{tabular}{lp{11.4cm}}
\toprule
Query & Description \\
\midrule
$w_a(x)$ & The $w_a$-weight of $x$. \\
$\depth{a}{x}$ & The $w_a$-weight of the nodes of the path from the root
			to $x$. \\
$\levelancestor{a}{x}{i}$ & The lowest ancestor $y$ of $x$ for which
			the $w_a$-weight of the nodes of the path from $y$ to $x$,
			excluding $x$, is at least $i$.\\
$\parent{a}{x}$ &  $\levelancestor{a}{x}{1}$.\\
$\Deg{a}{x}$ & The $w_a$-weight of the children of $x$. \\
$\childrank{a}{x}$ & The $w_a$ weight of $x$ and its left siblings. \\
$\childselect{a}{x}{i}$ & The leftmost child $y$ of $x$ for which
			$\childrank{a}{y} \geq i$. \\
$\numdescendants{a}{x}$ & The $w_a$-weight of the proper descendants of $x$. \\
$\bp{i}$ & The $i$-th character of $\BP$. \\
$\bpopen{x}$ & The index of the `1' character in $\BP$ that corresponds
			to $x$. \\
$\bpclose{x}$ & The index of the `0' character in $\BP$ that corresponds
			to $x$. \\
$\bpnode{i}$ & The node that corresponds to the $i$-th character of $\BP$. \\
$\bprank{a,b}{i}$ & The $w_a$-weight of the nodes that correspond
			to `1' characters of $\BP$ with index at most $i$,
			plus the $w_b$-weight of the nodes that correspond
			to `0' characters of $\BP$ with index at most $i$. \\
$\bpselect{a,b}{i}$ & The minimum index $j$ for which 
			$\bprank{a,b}{j} \geq i$. \\
\bottomrule
\end{tabular}
\end{table}

\begin{lemma}\label{lem:weighted-tree}
A weighted tree with $n$ nodes and $s=O(1)$ weight functions with weights
from $\range{0}{X-1}$, where $X=O(\log n)$, can be stored using at most
$2n\log(2X^s)+o(n)$ bits such
that the queries in Table~\ref{tab:weight-queries} are answered
in $O(1)$ time.
\end{lemma}
\begin{proof}
The representation is a variation of the unweighted tree representation
of Navarro and Sadakane~\cite{NavarroS14}.
We first review the latter representation.
Let $T$ be an unweighted tree, and let $P$ be its balanced parenthesis
representation.
Navarro and Sadakane showed that tree queries can be implemented by
supporting a set of \emph{base queries} that include
(1) the queries $\Degb$, $\childrankb$, and $\childselectb$
(2) the following queries on $P$ and a function
$f:\{0,1\}\to \{-1,0,1\}$ from a fixed set of functions $\mathcal{F}$:
\begin{align*}
\Sum{P}{f}{i}{j} & =\sum_{k=i}^{j}f(P[k])\\
\fwdsearch{P}{f}{i}{d} & =\min\{j\geq i:\Sum{P}{f}{i}{j} \geq d\}\\
\bwdsearch{P}{f}{i}{d} & =\max\{j\leq i:\Sum{P}{f}{i}{j} \geq d\}\\
\RMQi{P}{f}{i}{j} & =\mathrm{argmax}\{\Sum{P}{f}{1}{k}:i\leq k\leq j\}
\end{align*}
We note that Navarro and Sadakane used a slightly different definition for
$\fwdsearchb$ and $\bwdsearchb$, but their technique is easily modified for the
alternative definitions above.

The main idea of Navarro and Sadakane's representation is to partition $P$
into blocks of size $N=w^c$ for some constant $c$.
Each block of $P$ is stored using an aB-tree~\cite{Patrascu08} which
is able to support the base queries on the block in constant time.
The space of an aB-tree is $O(1)$ more than the information-theoretic lower
bound.
Since $P$ is a binary string, the space of one aB-tree is $N+O(1)$,
and the space for all trees is $2n+o(n)$.
In order to support base queries on the entire string $P$,
Navarro and Sadakane added additional data-structures.
The space of the additional data-structures is $O((n/N)\log^{O(1)}n)$.
Thus, by choosing large enough $c$, the additional space is $o(n)$.

For example, for supporting $\fwdsearchb$ for a function $f$,
a tree $T_f$ is constructed with weights from $\range{0}{N}$ on its edges,
and a weighted ancestor data-structure is built over $T_f$.
The tree $T_f$ is defined as follows.
Let $M_i$ be the maximum value of $\Sum{P}{f}{1}{j}$ for an
index $j$ that belongs to the $i$-th block of $P$.
The nodes of $T_f$ are $\{0,1,\ldots,n/N\}$.
A node $i>0$ is a child of node $j$, where $j$ is the minimum index for which
$j>i$ and $M_j > M_i$. The edge between these nodes have weight
$M_j-M_i$.
If no such index exists, $i$ is a child of node $0$.
A $\fwdsearch{P}{f}{i}{d}$ query is answered by first checking whether
the answer lies in the block of $i$ (using the aB-tree that stores the block).
If not, a weighted ancestor query
on $T_f$ finds the block in which the answer lies, and the location
inside the block is found using the aB-tree storing this block.

We now describe our representation for weighted trees.
Let $\BP$ denote the balanced parenthesis representation of $T$.
Define a string $P$ of length $2n$ in which each
character is a tuple of $s+1$ elements.
For an index $i$, the tuple $P[i]$ is $(w_1(v),w_2(v),\ldots,w_s(v),\BP[i])$,
where $v$ is the node that corresponds to $\BP[i]$.
As for the case of unweighted trees, it suffices to support
(1) the weighted tree queries
$\Degb$, $\childrankb$, and $\childselectb$
(2) $\Sumb$, $\fwdsearchb$, $\bwdsearchb$ and $\RMQib$ queries on the
string $P$ and a function $f$ (from a fixed set of functions $\mathcal{F}$)
which has the form
\[
\phi_{a,b}(x) = \begin{cases}
x_a & \text{if $x_{s+1} = 1$}\\
x_b & \text{if $x_{s+1} = 0$}
\end{cases}
\qquad \text{or} \qquad
\pi_{a}(x) = \begin{cases}
x_a & \text{if $x_{s+1} = 1$}\\
-x_a & \text{if $x_{s+1} = 0$}
\end{cases}
\]
where $x_a$ denotes the $a$-th coordinate of $x$
(we also denote $x_{\idxzero} = 0$).
To support the base queries, $P$ is partitioned into blocks and each block
is stored using an aB-tree.
The space of one aB-tree is now $N\log(2X^s)+O(1)$.
Thus, the total space of the aB-trees is $2n\log(2X^s)+o(n)$.
Additionally, since now the range of a function $f$ is $\range{-X}{X}$
whereas the range is $\range{-1}{1}$ for unweighted trees,
the space of the additional data-structures is increased
by a factor of at most $X$
(for example, the $T_f$ trees have now edge weights from $\range{0}{NX}$
and thus require more space).
Since $X=O(\log n)$, we can ensure the additional space is $o(n)$
by increasing $c$ by $1$.
\end{proof}

\subsection{Tree decomposition}
In the next lemma we present a tree decomposition that we will use in
our labeled tree representation.
The decomposition is slightly modified version of the decomposition of
Farzan and Munro~\cite{FarzanM08}.
An example for the decomposition is given in Figure~\ref{fig:decomposition-new}.
\begin{lemma}\label{lem:tree-cover}
For a tree $T$ with $n$ nodes and an integer $L$, there is a collection
$\D{T}{L}$ of subtrees of $T$ with
the following properties.
\begin{enumerate}
\item\label{enum:cover}
Every edge of $T$ appears in exactly one tree of $\D{T}{L}$.
\item\label{enum:size}
The size of each tree in $\D{T}{L}$ is at most $2L+1$.
\item\label{enum:number}
The number of trees in $\D{T}{L}$ is $O(n/L)$.
\item\label{enum:interval}
For every tree $T'\in  \D{T}{L}$ there are two intervals of integers
$I_1$ and $I_2$ such that a node $x\in T$ is
a non-root node of $T'$ if and only if $\preorderrankx{x} \in I_1 \cup I_2$,
where $\preorderrankx{x}$ is the preorder rank of $x$ in $T$.
The node with preorder rank $\min(\max(I_1) ,\max(I_2))$
is called the \emph{special node} of $T'$.
\item\label{enum:special}
For every $T'\in \D{T}{L}$, only the root and the special node of $T'$
can appear in other trees of $\D{T}{L}$.
\end{enumerate}
\end{lemma}
\begin{proof}
We first construct the decomposition of Farzan and Munro~\cite{FarzanM08}
(see Figure~\ref{fig:decomposition-old} for an example), which
satisfied the properties of the lemma, except for property~\ref{enum:cover}.
It also satisfies stronger versions of properties~\ref{enum:size}
and~\ref{enum:special}: Each tree in the decomposition has size at most $2L$,
and the common node of two trees in the decomposition can be only the root of
both trees.
Moreover, for a tree $T'$ in the decomposition,
for every edge $(v,w)$ between a node $v\in T'$ and a node $w\in T'$,
$v$ is the root of $T'$, except perhaps for one edge.

We now change the decomposition as follows
(see Figure~\ref{fig:decomposition-new}).
First, for every edge $(v,w)$ for which $v$ and $w$ are in different trees,
if $v$ is not a root of a tree in the decomposition,
add the node $w$ to the unique tree containing $v$.
Note that in this case, the tree also contains the predecessor of $v$ in the
preorder, so property~\ref{enum:interval} is maintained.
Otherwise, add a new tree to the decomposition that consists of the nodes
$v$ and $w$.
After the first step is performed, remove from the decomposition
all trees that consist of a single node.
It is easy to verify that the new decomposition satisfied all the properties
of the lemma.
\end{proof}
\begin{figure}
\centering
\hfill
\subfloat[\label{fig:decomposition-old}]{
\includegraphics[scale=0.5]{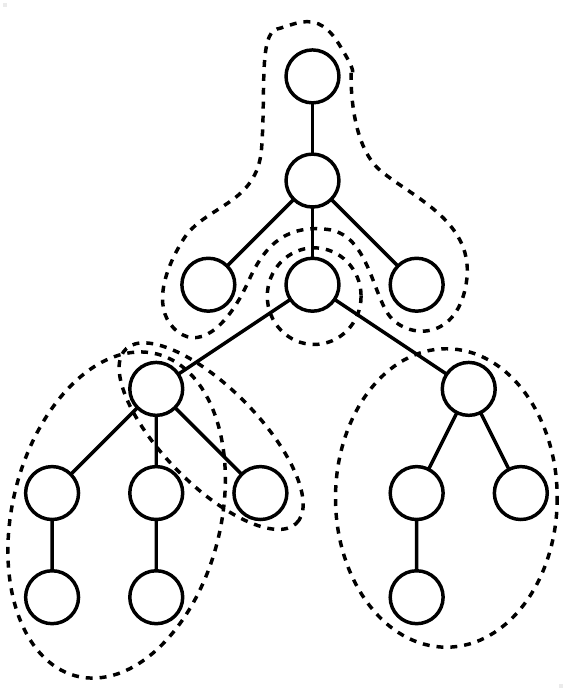}}
\hfill
\subfloat[\label{fig:decomposition-new}]{
\includegraphics[scale=0.5]{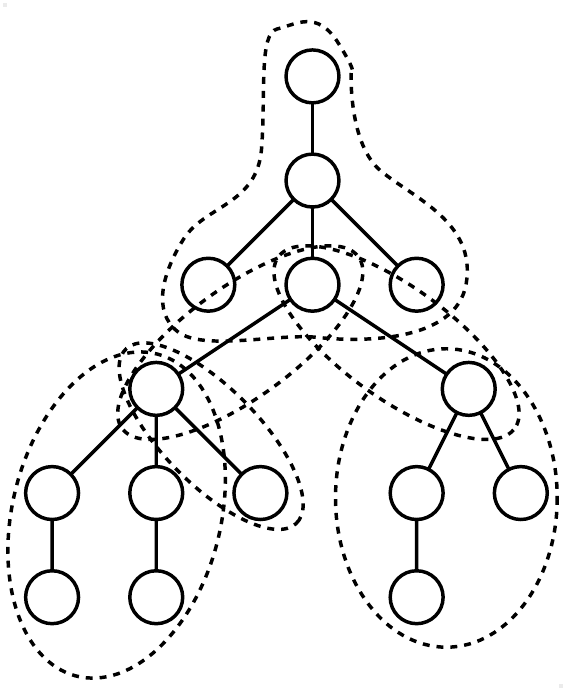}}
\hfill
\subfloat[\label{fig:decomposition-TL}]{
\includegraphics[scale=0.5]{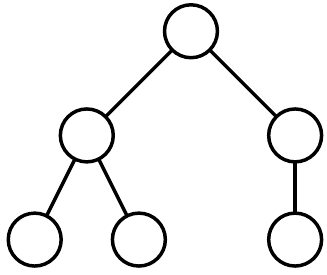}}
\hfill\hspace*{0pt}
\caption{Example of tree decomposition using $L=4$.
Figure~(a) shows the decomposition of Farzan and Munro,
and Figure~(b) shows are modified decomposition.
Figure~(c) shows the tree $\T{T}{L}$.
}
\end{figure}

For a tree $T$ and an integer $L$ we define a tree $\T{T}{L}$ as follows
(see Figure~\ref{fig:decomposition-TL}).
Construct a tree decomposition $\D{T}{L}$ according to
Lemma~\ref{lem:tree-cover}.
If the root $r$ of $T$ appears in several trees of $\D{T}{L}$,
add to $\D{T}{L}$ a tree that consists of $r$.
The set of nodes of $\T{T}{L}$ is the set $\D{T}{L}$.
For two trees $x',y'\in\D{T}{L}$, $x'$ is the parent of $y'$ in $\T{T}{L}$
if and only if the root of $y'$ is equal to the special node of $x'$
(or $x'$ is the tree that consists of $r$).

\section{Representation of labeled trees}
As in He et al.~\cite{HeMZ12}, we build trees $T^\alpha$ for every
$\alpha\in\range{1}{\sigma}$.
To build $T^\alpha$, we temporarily add to $T$ a new root with
label $\alpha$.
Then, let $X_\alpha$ be the set of all $\alpha$-nodes in $T$ and
their parents.
The nodes of $T^\alpha$ are the nodes of $X_\alpha$, and $x$ is a parent of
$y$ in $T^\alpha$ if and only if $x$ is the lowest (proper) ancestor of $y$
that appears in $X_\alpha$.
Unlike He et al., we do not store the tree $T^\alpha$.
Instead, we store a weighted tree that contains only part of the information
of $T^\alpha$.
The weighted tree is the tree $\T{T^\alpha}{L}$ obtained from
the tree decomposition of Lemma~\ref{lem:tree-cover}.
For the case $\sigma=w^{O(1)}$ the value of $L$ is
$L=f(n)$, where $f$ is a function that satisfies $f(n)=\omega(1)$,
and for large $\sigma$ set $L=\sqrt{\log\frac{\log\sigma}{\log w}}$.
Denote by $n_\alpha$ the number of $\alpha$-nodes in $T$.
We say that a character $\alpha$ is \emph{frequent} if $n_\alpha \geq L$.
We only construct the trees $\T{T^\alpha}{L}$ for frequent characters.

For a node $x'$ in $\T{T^\alpha}{L}$, let $\V{x'}$ be the set
of all $\alpha$-nodes in the tree $x'$, excluding the root of $x'$.
Let $\C{x'}$ be the leftmost $\alpha$-child of the root of $x'$
among the nodes of $\V{x'}$.
We define the following weight functions for $\T{T^\alpha}{L}$.
\begin{enumerate}
\item $w_{\idxtotal}(x')$ is the number of nodes in $\V{x'}$.
\item $w_{\idxle}(x')$ is the number of nodes in $\V{x'}$
whose preorder rank in $T^\alpha$ is
less than or equal to the preorder rank the special node of $x'$.
\item $w_{\idxgt}(x')$ is equal to $w_{\idxtotal}(x')-w_{\idxle}(x')$.
\item $w_{\idxeq}(x')$ is the number of nodes in $\V{x'}$
which are ancestors of the special node of $x'$.
\item $w_{\idxchildren}(x')$ is the number of $\alpha$-children of the
root of $x'$.
\item $w_{\idxleftmostchild}(x')$ is the rank of $\C{x'}$ among the nodes
of $\V{x'}$, when the nodes are sorted according to preorder
(if $\C{x'}$ does not exist, $w_{\idxleftmostchild}(x') = 0$).
\item $w_{\idxspecial}(x')$ is equal to $1$ if the special node of
$x'$ is an $\alpha$-node, and $0$ otherwise.
\end{enumerate}

Our representation of $T$ consists of the following data-structures.
We store a rank-select structure on $P_T$, using one of the structures of
Theorem~\ref{thm:rank-select} according to the size of the alphabet.
We also store an unlabeled tree $\hat{T}$ obtained from $T$ by
removing the labels.
$\hat{T}$ is stored using Theorem~\ref{thm:unlabeled-tree}.
We also store the trees $\T{T^\alpha}{L}$.
In order to reduce the space, we do not store these trees individually.
Instead, we merge them into a single tree $T'$.
The tree $T'$ contains a new root node, on which the trees $\T{T^\alpha}{L}$
are hanged, ordered by increasing values of $\alpha$.
The tree $T'$ is stored using the representation of
Lemma~\ref{lem:weighted-tree}.
In order to map a node of $\T{T^\alpha}{L}$
to the corresponding node of $T'$, and vice versa, we store
the rank-select structure of Theorem~\ref{thm:rank-select-binary}
on the string $N=10^{n_{\alpha_1}}10^{n_{\alpha_2}}\cdots$, where
$\alpha_1 < \alpha_2 < \cdots$ are the frequent characters.
We also store the rank-select structure of Theorem~\ref{thm:rank-select-binary}
on a binary string $F$ of length $\sigma$ in which $F[\alpha]=1$ if
$\alpha$ is frequent.
Since a mapping between nodes of $\T{T^\alpha}{L}$ and nodes of $T'$ can
be computed in constant time, in the following we shall assume that the trees
$\T{T^\alpha}{L}$ are available.

We now analyze the space complexity of the representation.
The space for $P_T$ is $n H_0(P_T)+o(n)$ bits for small alphabet,
and $n H_0(P_T)+o(n H_0(P_T))+o(n)$ bits for large alphabet.
The space for $\hat{T}$ is $2n+o(n)$ bits.
The tree $T'$ has $O(n/L)$ nodes, and the weight functions have ranges
$\range{0}{L}$. Thus, the space for $T'$ is $O((n/L)\log L)+o(n)=o(n)$ bits.
The strings $N$ and $F$ have at most $n$ zeros and at most $n/L$ ones.
Thus, the space for the rank-select structures on these strings is
$O((n/L)\log L)+o(n)=o(n)$ bits.
Therefore, the total space of the representation
is $n H_0(P_T)+2n+o(n)$ for small alphabet,
and $n H_0(P_T)+2n+o(n H_0(P_T))+o(n)$ for large alphabet.

In the following, when we use an unlabeled tree operation
(e.g., $\parentx{x}$) we assume that the operation is performed on
$\hat{T}$, and when we use a weighted tree operation  we assume that the
operation is performed on $\T{T^\alpha}{L}$.

\subsection{Mapping from $\T{T^\alpha}{L}$ to $T$}

Let $x'$ be a node of $ \T{T^\alpha}{L}$.
From property~\ref{enum:interval} of Lemma~\ref{lem:tree-cover} there are
two intervals $[l_1(x'),r_1(x')]$ and $[l_2(x'),r_2(x')]$
such that an $\alpha$-node $x$ of $T$ is a non-root node of $x'$
if and only if the rank of $x$ among all the $\alpha$-nodes of $T$,
sorted in preorder, is in one of the intervals.
These intervals can be computed as follows:
\begin{align*}
r_1(x') & = \bprank{\idxle,\idxgt}{\bpopen{x'}}\\
l_1(x') & = r_1(x')-w_{\idxle}(x')+1\\
r_2(x') & = \bprank{\idxle,\idxgt}{\bpclose{x'}}\\
l_2(x') & = r_2(x')-w_{\idxgt}(x')+1
\end{align*}
The set $\V{x'}$ can be computed with
\[ V(x') = \{\preorderselectx{\select{\alpha}{P_T}{s}}
  : s \in [l_1(x'),r_1(x')]\cup[l_2(x'),r_2(x')] \}, \]
and $\C{x'}$ can be computed with $\preorderselectx{\select{\alpha}{P_T}{k}}$,
where $k = l_1(x')+w_{\idxleftmostchild}(x')-1$ if
$w_{\idxleftmostchild}(x') \leq w_{\idxle}(x')$,
and $k = l_2(x')+w_{\idxleftmostchild}(x')-w_{\idxle}(x')-1$ otherwise.

\subsection{Mapping from $T$ to $\T{T^\alpha}{L}$}
Let $x$ be a node of $T$ and $\alpha$ be a frequent character.
If $x\in X_\alpha$, we want to compute the node $x'\in \T{T}{L}$ such
that $x\in \V{x'}$.
We denote this node by $\map{x}$.
Additionally, we compute whether $x$ is the special node of $\map{x}$.

We consider two cases. The first case is when $x$ is an $\alpha$-node.
Then,
\[ \map{x}=\bpnode{\bpselect{\idxle,\idxgt}{\preorderrank{\alpha}{x}}}. \]
Moreover, $x$ is the special node of $\map{x}$ if and only if
$w_{\idxspecial}(x')=1$ and $\preorderrank{\alpha}{x}=r_1(x')$.

Now suppose that $x$ is not an $\alpha$-node.
The algorithm for computing $\map{x}$ is as follows.
\begin{enumerate}
\item
Let $u$ and $v$ be the $\alpha$-descendants of $x$ with minimum and maximum
preorder ranks, respectively
($u,v$ are computed using the structures on $P_T$ and $\hat{T}$).
If $u,v$ do not exist return $\NULL$.
\item
Compute $u'=\map{u}$ and $v'=\map{v}$.
\item
If $x\neq \lcax{u}{v}$
\begin{enumerate}
\item
If $\parentx{u} \neq x$ return $\NULL$.
\item
Compute $\C{u'}$.
If $\C{u'}=u$ then
$\map{x}=\parentx{u'}$ and $x$ is the special node of $\map{x}$.
Otherwise, $\map{x}=u'$ and $x$ is not the special node.
\label{alg:x-not-lca}
\end{enumerate}
\item
If $x=\lcax{u}{v}$
\begin{enumerate}
\item
If $u'$ is an ancestor of $v'$ (including $u'=v'$), compute $\V{u'}$.
Scan the nodes of $\V{u'}$ and check for each node whether it is a child of $x$.
If no child of $x$ was found, return $\NULL$.
Compute $\C{u'}$. If $\C{u'}$ exists and $\parentx{\C{u'}}=x$
then $\map{x}=\parentx{u'}$ and $x$ is the special node of $\map{x}$.
Otherwise, $\map{x}=u'$ and $x$ is not the special node.
\item If $v'$ is a proper ancestor of $u'$, 
handle this case analogously to the handling of the previous case.
\item If neither node $u'$ or $v'$ is an ancestor of the other,
$\map{x}=\lcax{u'}{v'}$ and $x$ is the special node of $\map{x}$.
\end{enumerate}
\end{enumerate}
To see the correctness of the algorithm above, observe that
if $u,v$ do not exist then $x$ does not have $\alpha$-children.
Since $x$ is not an $\alpha$-node, by definition $x \notin T^\alpha$.
Now suppose that $u,v$ exist.
If $x \neq \lcax{u}{v}$, the only possible $\alpha$-child of $x$ is $u$.
Thus, if $\parentx{u}\neq x$ then $x\notin T^\alpha$.
If $\parentx{u} = x$ then $x$ is a node of $T^\alpha$.
By the definition of $\map{\cdot}$, $u$ is not the root of $u'$,
and therefore the node $x$ is also a node of $u'$.
Thus, $\map{x}=\parentx{u'}$ if $x$ is the root of $u'$, and
$\map{x}=u'$ otherwise.
Since $u$ is the only $\alpha$-child of $x$, it follows that the former case
occurs if and only if $\C{u'}=u$.

Now consider the case $x = \lcax{u}{v}$.
Denote by $u_2$ and $v_2$ the ancestors of $u$ and $v$ that are
children of $x$ in $T^\alpha$.
By the definition of $u$ and $v$, if $x$ has $\alpha$-children then
every $\alpha$-child $y$ of $x$ is between $u_2$ and $v_2$ in $T^\alpha$.
If  $u'$ is an ancestor of $v'$ then $u_2$ and $v_2$ must belong to the
tree $u'$.
Thus, all the $\alpha$-children of $x$ (if there are any) belong to $u'$.
Therefore, it suffices to scan the nodes of $\V{u'}$ in order to decide
whether $x$ has $\alpha$-children.
If $x$ has $\alpha$-children then these children are non-root nodes of $u'$.
It follows that $x$ belongs to the tree $u'$.
Moreover,
$\map{x}=\parentx{u'}$ if $x$ is the root of $u'$, and
$\map{x}=u'$ otherwise.

\subsection{Answering queries}
In this section describe how to implement the queries of
Table~\ref{tab:labeled-queries}.
The queries $\Labelb$, $\preorderrankb$, $\preorderselectb$,
$\numdescendantsb$ and $\postorderrankb$ are answered as in~\cite{HeMZ12}.
Answering the remaining queries is also similar to~\cite{HeMZ12},
but here additional steps are required as a weighted tree $\T{T^\alpha}{L}$
holds less information than $T^\alpha$.
The general idea is to use the tree $\T{T^\alpha}{L}$ to get an approximate
answer to the query. Then, by enumerating the nodes of constant number
of $\V{\cdot}$ set, the exact answer is found.
The time complexity of answering a query is thus $O(L\cdot\tselect)$,
where $\tselect$ is the time for a select query on $P_T$.
Since $\tselect = O(1)$ for small alphabet
and $\tselect = o(\sqrt{\log\frac{\log\sigma}{\log w}})$ for large alphabet,
it follows that the time for answering a query is $\omega(1)$ for small alphabet
and $O(\log\frac{\log\sigma}{\log w})$ for large alphabet.

We assume that $\alpha$ is a frequent character until the end of the section.
Handling queries in which $\alpha$ is non-frequent is done by enumerating all
the $\alpha$-nodes in $T$.

\subsubsection{$\parent{\alpha}{x}$ query}
We consider two cases.
If $x$ is an $\alpha$-node the query is answered as follows.
\begin{enumerate}
\item Compute $x'=\map{x}$.
\item Compute $\V{x'}$. Scan the nodes of $\V{x'}$ in reverse order,
and check for each node $v$ whether $v$ is an ancestor of $x$.
If an ancestor of $x$ is found, return it.
\item
Compute $y' = \parent{\idxeq}{x'}$.
If $y'$ does not exist return $\NULL$.
\item
Compute $\V{y'}$. Scan the nodes of $\V{y'}$ in reverse order,
and check for each node $v$ whether $v$ is an ancestor of $x$.
When an ancestor of $x$ is found, return it.
\end{enumerate}
If $x$ is not an $\alpha$-node then the query is answered as
in He et al.~\cite{HeMZ12}.


\subsubsection{$\depth{\alpha}{x}$ query}
\begin{enumerate}
\item
Let $y=x$ if $x$ is an $\alpha$-node, and $y=\parent{\alpha}{x}$ otherwise.
\item
Compute $y' = \map{y}$.
\item
Compute $\V{y'}$ and scan the nodes of $\V{y'}$.
Count the number of nodes that are ancestors of $x$,
and let $i$ denote this number.
\item
Return $i + \depth{\idxeq}{y'}-w_{\idxeq}(y')$.
\end{enumerate}

\subsubsection{$\levelancestor{\alpha}{x}{i}$ query}
\begin{enumerate}
\item
Compute $y = \parent{\alpha}{x}$.
If $y$ does not exist return $\NULL$.
\item
If $i=1$ return $y$.
\item
Compute $y' = \map{y}$.
\item
Compute $\V{y'}$. Scan the nodes of $\V{y'}$ in reverse order.
For each node $v$, if $v$ is an ancestor of $x$, decrease $i$ by $1$.
If $i$ becomes $0$ return $v$, and otherwise continue with the scan.
\item
Compute $z' = \levelancestor{\idxeq}{y'}{i}$.
If $z'$ does not exist return $\NULL$.
\item
Set $i \gets i-(\depth{\idxeq}{\parentx{y'}}-\depth{\idxeq}{z'})$.
\item
Compute $\V{z'}$. Scan the nodes of $\V{z'}$ in reverse order.
For each node $v$, check whether $v$ is an ancestor of $x$.
If $v$ is an ancestor, decrease $i$ by $1$.
If $i$ becomes $0$ return $v$, and otherwise continue with the scan.
\end{enumerate}

\subsubsection{$\Deg{\alpha}{x}$ query}
\begin{enumerate}
\item
Compute $x' = \map{x}$.
If $x'$ does not exist return $0$.
\item
If $x$ is not the special node of $x'$,
compute $\V{x'}$ and scan the nodes of $\V{x'}$.
Return the number of nodes that are children of $x$.
\item
Return $\Deg{\idxchildren}{x'}$.
\end{enumerate}

We now explain the correctness of the algorithm above.
If $x$ is not the special node of $x'$ then all the children of $x$
in $T^\alpha$ are in the tree $x'$. Thus, it suffices to scan $\V{x'}$.
If $x$ is the special node, the set of $\alpha$-children of $x$ is
precisely the set of $\alpha$-children of the roots of all trees $y'$
such that $y'$ is a child of $x'$ in $\T{T^\alpha}{L}$.
Therefore, $\Deg{\alpha}{x}=\Deg{\idxchildren}{x'}$.

\subsubsection{$\childrank{\alpha}{x}$ query}
If $x$ is an $\alpha$-node the query is answered as follows.
\begin{enumerate}
\item
Compute $x' = \map{x}$.
\item
Compute $\V{x'}$ and scan the nodes of $\V{x'}$.
Count the number of nodes that are left sibling of $x$,
and let $i$ denote this number.
\item
Compute $u = \parentx{x}$ and $u' = \map{u}$.
\item
If $u$ is not the special node of $u'$ return $i$.
\item
Return $i + \childrank{\idxchildren}{x'}-w_{\idxchildren}(x')$.
\end{enumerate}
We now consider the case when $x$ is not an $\alpha$-node.
\begin{enumerate}
\item
Compute $u = \parentx{x}$ and $u' = \map{u}$.
If $u'$ does not exist return $0$.
\item
If $u$ is not the special node of $u'$,
compute $\V{u'}$ and scan the nodes of $\V{u'}$.
Return the number of nodes that are left sibling of $x$.
\item
Let $v$ be the $\alpha$-predecessor.
If $v$ does not exist or if $\preorderrankx{v} \leq \preorderrankx{u}$
return $0$.
\item
Compute $v'=\map{v}$.
\item
Let $w'$ be the child of $u'$ which is an ancestor of $v'$.
\item
Compute $\V{w'}$ and scan the nodes of $\V{w'}$.
Count the number of nodes that are left sibling of $x$,
and let $i$ denote this number.
\item
Return $i + \childrank{\idxchildren}{x'}$.
\end{enumerate}

\subsubsection{$\childselect{\alpha}{x}{i}$ query}
\begin{enumerate}
\item
Compute $x' = \map{x}$.
If $x$ does not exist return $\NULL$.
\item
If $x$ is not the special node of $x$, 
compute $\V{x'}$ and scan the nodes of $\V{x'}$.
For each node $v$, check whether $v$ is a child of $x$.
Return the $i$-th child.
\item
Compute $y'=\childselect{\idxchildren}{x'}{i}$.
\item
Set $i\gets i-(\childrank{\idxchildren}{y'}-w_{\idxchildren}(y'))$.
\item
Compute $\V{y'}$ and scan the nodes of $\V{y'}$.
For each node $v$, check whether $v$ is a child of $x$.
Return the $i$-th child.
\end{enumerate}

\subsubsection{$\postorderselect{\alpha}{i}$ query}
\begin{enumerate}
\item
Return $\bpnode{\bpselect{\idxzero,\idxtotal}{i}}$.
\end{enumerate}


\bibliographystyle{plain}
\bibliography{string-index}
\end{document}